\documentclass[11pt]{article}
\usepackage[utf8]{inputenc}

\usepackage{amsmath}
\usepackage{amssymb}
\usepackage{amsthm}
\usepackage{authblk}
\usepackage[mode=buildnew,subpreambles=true]{standalone}

\usepackage{algpseudocode}
\usepackage{algorithm}
\usepackage{tikz}
\usetikzlibrary{positioning}
\usepackage{fullpage}
\usepackage{xcolor}

\usepackage{hyperref}

\title{A Critique of Sopin's ``$\ph = \pspace$''\thanks{Supported in part by NSF grant 
		CCF-2006496.}}

\author{Michael C. Chavrimootoo} 
\author{Ian Clingerman}
\author{Quan Luu}
\affil{Department of Computer Science\\University of Rochester\\Rochester, NY 14627, USA}

\newcommand{\condition}{\,\mid \:}
\newcommand{\naturalnumber}{\ensuremath{{\mathbb{N}}}}
\newcommand{\naturalnumberpositive}{\ensuremath{{\mathbb{N}^+}}}

\newcommand{\p}{\ensuremath{{\rm P}}}
\newcommand{\np}{\ensuremath{{\rm NP}}}

\newcommand{\ph}{\ensuremath{{\rm PH}}}
\newcommand{\pspace}{\ensuremath{{\rm PSPACE}}}
\newcommand{\dspace}{\ensuremath{{\rm DSPACE}}}
\newcommand{\lspace}{\ensuremath{{\rm L}}}
\newcommand{\sat}{\ensuremath{{\rm SAT}}}

\newcommand{\qbf}{\ensuremath{{\rm QBF}}}

\newtheorem{theorem}{Theorem}
\newtheorem{proposition}[theorem]{Proposition}

\date{December 9, 2022}

\begin{document}\sloppy

\maketitle

\begin{abstract}
We critique Valerii Sopin's paper ``$\ph=\pspace$''~\cite{sop:t:phpspace}.
The paper claims to resolve one of the major open problems of theoretical computer science by leveraging the Skolemization of existential quantifiers of quantified boolean formulas
to show that $\qbf$ (a well-known $\pspace$-complete problem) is
in $\Pi_4^p$, and thus $\ph=\pspace$.
In this critique, we highlight problems in that paper
and conclude that it fails to establish that $\ph=\pspace$.
\end{abstract}

\section{Introduction}
In this paper, we critique Valerii Sopin’s paper ``$\ph=\pspace$''~\cite{sop:t:phpspace}. The paper introduces two theorems. The first theorem is about Skolemization of quantified boolean formulas, while the second theorem attempts to leverage the first theorem to collapse the polynomial hierarchy and prove that $\ph = \pspace$. Before going into more detail, we first discuss the importance of such a result.

In computational complexity theory, the polynomial hierarchy is defined as 
$\ph=\p \cup \np \cup \np^{\np} \cup \np^{\np^{\np}} \cup \cdots$, and $\pspace$ refers to the class of problems that can be solved using at most polynomial space. Both are fundamentally important and intensely studied classes. It is known that $\ph \subseteq \pspace$. However, whether $\ph$ equals $\pspace$ to this day remains one of the central open problems in theoretical computer science. Informally, this is because most known techniques cannot settle
this question.\footnote{
This is because a proof resolving $\ph$ vs.\ $\pspace$ cannot relativize, as follows from the results of Baker, Gill, and Solovay~\cite{bak-gil-sol:j:rel} and Yao~\cite{yao:c:separating}. 
}
There are many consequences to proving $\ph=\pspace$~\cite{far:url-rm-font:ph-equals-pspace-consequences}.
Hypothetically, if one were able to show that $\pspace$ is no more complex than some level $k$ of $\ph$, it would indicate that all levels above $k$ collapse to the $k$th level (since they are all contained in $\pspace$). Even for a large value of $k$, a collapse would help advance research into other open questions of complexity theory. For example, $\ph = \pspace$ would yield the existence of $\ph$-complete problems, a currently unresolved issue (that is in fact equivalent to the issue of whether \ph\ collapses, i.e., for some $i$, $\ph = \Sigma_{i}^{p} $). Another serious implication of this result would be resolving the relationship between $\lspace$ and classes like $\ph$ or $\np.$\footnote{$\lspace$ is the class of decision problems solvable by a deterministic Turing machine in logarithmic space. It is known that $\lspace \subseteq \p$. Thus if $\lspace = \np$, then $\lspace = \p = \np = \ph$. It is also true by the space hierarchy theorem that $\lspace \neq \pspace$. Consequently, if $\ph = \pspace$, then $\lspace \neq \ph$, and thus $\lspace \neq \np$.}

The approach used in Sopin's paper to purportedly prove its theorems is as follows: first apply Skolemization to an arbitrary quantified boolean formula and then construct an algorithm to solve $\qbf$ from one that solves $\Pi_4\sat$. However, we will show how the proofs for these theorems are flawed.

In Section~\ref{s:prelims}, we introduce the definitions and notations used within Sopin's proofs. In Sections~\ref{s:theorem-one} and~\ref{s:theorem-two}, we summarize the key points of Theorems~1 and~2 and analyze each  proof. We will use Section~\ref{s:skolem} to point out a key observation on the use of Skolemization in the context of the paper. Finally, in Section~\ref{s:conclusion}, we conclude our critique.

\section{Preliminaries}\label{s:prelims}

We present here many standard concepts in complexity theory. Equivalent definitions can be found in most modern 
textbooks on complexity~\cite{hem-ogi:b:companion,aro-bar:b:complexity,sip:b:introduction-third-edition}.

As to notation,
given a string $w$, let $|w|$ denote the length of $w$.
Additionally, let 
$\naturalnumber = \{0, 1, 2, \ldots\}$ and
let $\naturalnumberpositive = \{1, 2, 3, \ldots\}$.

Throughout this paper, we will speak of boolean formulas and boolean variables. We will assume that 
our (boolean) variables are always assigned~1 (true) or~0 (false). 
A boolean formula $\phi$ with boolean variables $x_1, \ldots, x_n$ is denoted by $\phi (x_1, \ldots, x_n)$.
Quantified boolean formulas are of the form $(Q_1 x_1)\cdots(Q_n x_n)[\phi (x_1,\ldots, x_n)]$, such that for each $i \in \{1, \ldots, n\}$,
it holds that $Q_i$ is either $\exists$ or $\forall$.

Let us now define the polynomial hierarchy. 
Fix $i \geq 1$. A language $L$ is in $\Sigma_i^p$ exactly if there is a polynomial $q$ and a polynomial-time computable predicate $R$ such that for all $x$ it holds that
$$x\in L \iff (\exists w_1: |w_1| \leq q(|x|))(\forall w_2: |w_2| \leq q(|x|)\cdots
(Q_iw_i: |w_i| \leq q(|x|))[R(x, w_1, \ldots, w_i)],$$
where $Q_i$ is $\exists$ if $i$ is odd and $\forall$ if $i$ is even.
Additionally, a language $L'$ is in $\Pi_i^p$ exactly if $\overline{L'} \in \Sigma_i^p$.
The polynomial hierarchy is defined as the class $\ph = \bigcup_{i\in\naturalnumberpositive}\Sigma_i^p$.\footnote{Readers familiar with the concepts will notice that we omitted explicit mention of $\Sigma_0^p$, but it is of course contained in $\Sigma_1^p$ and thus not actually omitted. This ``omission'' occurs simply because we do not need to refer to that class in the rest of this paper, and so it saves us the trouble of defining it.} This definition yields the same classes as the other common definition, drawn on in our introduction, i.e., $\Sigma_3^p = \np^{\np^\np}$. 

For each of the abovementioned classes (except $\ph$), we can define the following
canonical complete problems.
The definition that follows is adapted from one given by Arora and Barak~\cite{aro-bar:b:complexity}.
For each $i \geq 1$, the problem
$\Sigma_i\sat$, which is complete for $\Sigma_i^p$, is defined as the set of quantified boolean formulas of the
form
$(\exists u_1)\cdots(Q_iu_i)[\psi(u_1, \ldots, u_i)]$ that are true, 
where $\psi$ is a boolean formula, each of $u_1, u_2, \ldots, u_i$ denotes a list
of a boolean variables, and $Q_i$ is $\exists$ if $i$ is odd and $\forall$
if $i$ is even. The problem $\Pi_i\sat$, which is complete for $\Pi_i^p$, is defined
analogously, except that $Q_i$ is $\forall$ when $i$ is odd, and $\exists$ if $i$ is even.

$\pspace$ is the class of languages accepted by a Turing machine that runs in polynomial space. 
Formally, $\pspace = \bigcup_{k \in \naturalnumberpositive} \dspace[n^k]$.
Additionally, ``$\pspace$ embodies the power of polynomially bounded quantifiers"~\cite{hem-ogi:b:companion}. 
And so, $\qbf$, the set of all true quantified boolean formulas, is a canonical $\pspace$-complete 
problem. 
It's easy to see that $\qbf = \bigcup_{i \in \naturalnumber} \Sigma_i\sat \cup \Pi_i\sat$, and thus
$\ph \subseteq \pspace$.

Skolemization is the process of removing existentially quantified variables 
from a quantified boolean formula, and replacing them with Skolem constants or Skolem functions, 
which are boolean functions that are only over the universally quantified variables that appear 
before the existentially quantified variable being removed in the quantifier order. Interested 
readers may consult the textbook by Chang and Lee~\cite{cha-lee:b:symbolic-logic} for more details.

\section{On Theorem~1}\label{s:theorem-one}
We will now focus on Theorem~1 of Sopin's paper. The theorem has been reformulated for the sake of clarity.  
\begin{theorem}[\cite{sop:t:phpspace}]
     For an arbitrary $n \in \naturalnumberpositive$,
     let $\Phi = (\Omega_1x_1)(\Omega_2x_2)\cdots(\Omega_nx_n)[\phi(x_1,\ldots,x_n)]$ be an arbitrary boolean formula where 
     $(\Omega_1, \ldots, \Omega_n) \in \{\exists, \forall\}^n$.
     Let $I = \{i \condition i \leq n \land \Omega_i = \exists\}$.
     Then, $\Phi$
     is a true quantified boolean formula if and only 
     for every $i \in I$, there
     is a boolean function $y_i$ over the variables in
     $\{x_j \condition j < i \land \Omega_j = \forall\}$
     such that the quantified boolean formula that results from substituting 
     each $x_i$
     with $y_i$ is a tautology.
\end{theorem}

To the best of our understanding, the theorem is essentially proposing that Skolemization preserves 
satisfiability, since the process described by the theorem resembles that of Skolemization. 
However, it is well-known that Skolemization does in fact preserve satisfiability (for reference, see textbooks by Chang and Lee~\cite{cha-lee:b:symbolic-logic} and Loveland~\cite{lov:b:automated-theorem-proving}). 
In order to leverage this technique to decide \qbf, it follows that one would at the very least need a polynomial upper bound on the space complexity of Skolemization. 
Unfortunately, no such result is currently known. We elaborate this point further in Section~\ref{s:skolem}.
Additionally, we notice several issues with the given proof of Theorem~1 and discuss them in the 
rest of this section.

In their proof of Theorem~1, the author introduces a recursive algorithm 
to test for membership in $\qbf$.
It works by removing quantifiers sequentially and producing, based on the removed quantifier, a new and logically equivalent formula. After that, the author adds a note on how ``a [boolean] function determines the truth table''~\cite{sop:t:phpspace} of its variables. The proof then follows with two examples to help illustrate its main argument.

First, we comment on the structure of the proof. It is unclear how each part of the proof relates to the other and to the main argument. Therefore, throughout our analysis, we will attempt to understand the role of each part of the proof. 

Sopin presents a recursive algorithm to decide $\qbf$ and claims that the proof of
Theorem~1 follows from it, however this algorithm is not well-defined. The recursive step is the only defined aspect of the algorithm: It involves pulling off the first quantified variable and checking both values for that variable. Depending on the quantifier, a new equation is formed using these two equations with both possible values of first variable. In the definition of this algorithm, there is neither a base case nor a recursive call. A base case is needed in order to show when the algorithm terminates and a recursive call is needed in order to show where the recursive step is applied. So it is ambiguous as to whether the algorithm presented is indeed recursive, or if only the first quantifier is to be pulled off. It appears to us that is a common algorithm to decide
$\qbf$, but it is unclear how the proof of Theorem~1 follows from it.

We would now like to turn our attention to a statement made at the end of the algorithm's 
description:~``Notice that a [boolean] function determines the truth table (one-to-one 
correspondence)"~\cite{sop:t:phpspace}. The term ``truth table" in this context is ambiguous. If 
the intention of the paper is indeed to use truth-tables to represent Skolem functions, as in 
Example~1 of Theorem~1 (discussed later), then this is a matter of concern, as the size of a 
truth-table is exponential in the number of its variables.
On the other hand, it's possible that the quoted statement is simply a fact of the relationship between boolean functions and truth-tables. Still, the paper fails to show the potential space complexity of this transformation, which is crucial when dealing with $\pspace$-complete problems. 

At the end of the proof of Theorem~1, two examples are presented in order to add clarity to the 
theorem. 
Example~1 provides an explicit description of a Skolem function, where existentially quantified variables can be rewritten as functions of the variables that come before it in the formula. At the end of the example, it is mentioned that the Skolem function ``is indeed the truth table, where values of $x_1, \ldots, x_{k-1}$ determine the value of $x_k$"~\cite{sop:t:phpspace}. As mentioned before, the use of truth-table in this context would lead to an exponential size blow-up. Example~2 is just a restatement of the theorem using an explicit quantified boolean formula. Example~2 states that ``$\forall x_1 \exists z_1 \forall x_2 \exists z_2 \forall x_3 \exists z_3$ is a true [quantified boolean formula] if and only if there exist such boolean functions $y_1 \colon \{0, 1\} \rightarrow \{0, 1\}, \: y_2 \colon \{0, 1\}^2 \rightarrow \{0, 1\}, \: y_3 \colon \{0, 1\}^3 \rightarrow \{0, 1\}$ that $\phi (x_1, y_1 (x_1), x_2, y_2 (x_1, x_2), x_3, y_3 (x_1, x_2, x_3))$ is tautology"~\cite{sop:t:phpspace}. This example just shows how a formula looks after Skolemization has occurred. The formula must be a tautology in order to be a true quantified boolean formula because only universally quantified variables are left. 
The  
replacement of existentially quantified variables with boolean functions over the correct variables
appears to have been performed correctly but there is neither a 
proof that these functions are Skolem functions 
nor a direct connection to the proof of Theorem~1.

Thus while the statement is true, the proof given does not support the theorem. 
It's unclear, based on our observations about Skolemization not being known to be computable using polynomial space, how this theorem will become useful in the rest of Sopin's paper (and indeed, we discuss in the next section that we see no such connection).

\section{On Theorem~2}\label{s:theorem-two}

In this section, we will focus our attention to the following theorem of Sopin's paper.

\begin{theorem}[\cite{sop:t:phpspace}]
$\Pi_4^p = \pspace$.
\end{theorem}

As one would expect, the purported proof takes an arbitrary quantified boolean formula $\Phi$
and, from it, constructs a quantified boolean formula $\Phi'$ (which has a specific syntactic form 
that we specify later in this section) in polynomial time, such that
$\Phi \in \qbf \iff \Phi' \in \Pi_4\sat$.
Let us preface that the proof presented in the paper is often unclear about what it
means and makes several logical leaps, and so we approach the arguments in the proof using our best understanding of what the author could have meant.

Our first comment is about clarity and touches on the form of the quantified boolean formulas. 
The author assumes that quantified boolean formulas are of the form (which we will often refer to as the ``standard'' form in this critique)
$$(\forall x_1)(\exists y_1)\cdots(\forall x_n)(\exists y_n)[\phi(x_1, y_1, \ldots, x_n, y_n)],$$
for some $n \in \naturalnumberpositive$,
which at a glance seems incorrect. 
For example, the formula $(\exists x, y)(\forall z)[(x\lor y\lor z)]$ is certainly in $\qbf$,
but is not included in the paper's treatment. Indeed, the general form given seems to miss formulas with an
odd number of variables, formulas that start with an existential quantifier, and formulas that have multiple variables
bound to the same quantifier. However, what the paper does not make clear, is that all such formulas can be converted to the ``standard'' form in polynomial-time. 
The key to putting arbitrary quantified boolean formulas
into ``standard'' form is the use of dummy variables.\footnote{
For our purposes, a dummy variable is one that is quantified over, but does not appear in the boolean formula. That is certainly legal. For example, we can say that the formula $(\forall x)(\exists y)(\forall z)[(x \lor y)]$ is over the set of variables $\{x, y, z\}$. In this case, $z$ is a dummy variable.}
Our example,
$$(\exists x, y)(\forall z)[(x\lor y\lor z)],$$
introduced earlier in this paragraph exhibits all the issues that seem to exist with the ``standard'' form. 
And so, we shall present, in an informal manner (since the issue is rather simple and does
not warrant much more than an example),
how to convert the above formula to ``standard'' form.
Let us first make the first quantifier be a universal quantifier by introducing the
dummy variable $\alpha$. 
This yields the formula
$$(\forall \alpha)(\exists x, y)(\forall z)[(x\lor y\lor z)].$$
Next, we separate the variables $x$ and $y$ so that each quantifier is only bound to one variable by introducing the dummy variable $\beta$. This yields the formula
$$(\forall \alpha)(\exists x)(\forall \beta)(\exists y)(\forall z)[(x\lor y\lor z)].$$
To finish, since we need an even number of variables,  we introduce the dummy variable
$\gamma$ and obtain
$$(\forall \alpha)(\exists x)(\forall \beta)(\exists y)(\forall z)(\exists \gamma)[(x\lor y\lor z)].$$
It is not hard to see that if the original formula has $n$ variables, then the
new formula will have at most $2n+2$ variables.

We will now focus on the correctness of 
the proof of Theorem~2.
For the rest of this section, we will fix, for some $n \in \naturalnumberpositive$
and some boolean formula (with no quantifiers) $\phi$ that is over $2n$ variables,
the following 
quantified boolean formula 
$$\Phi = (\forall x_1)(\exists y_1)\cdots(\forall x_n)(\exists y_n)\allowbreak[\phi(x_1, y_1, \ldots, x_n, y_n)].$$

The paper seeks to 
construct, from $\Phi$, 
the following (purportedly logically equivalent) formula (which is not in ``standard'' form)
\begin{align*}
\Phi' = &(\forall x_1, \ldots, x_n)(\exists y_1, \ldots, y_n)[
\phi(x_1, y_1, \ldots, x_n, y_n) \land\\
&(\forall \hat{x}_n)(\exists z_n)[\phi(x_1, y_1, \ldots, x_{n-1}, y_{n-1}, \hat{x}_n, z_n)]
\land\\
&(\forall \hat{x}_{n-1}, \hat{x}_n)(\exists z_{n-1}, z_n)
[\phi(x_1, y_1, \ldots, x_{n-2}, y_{n-2}, \hat{x}_{n-1}, z_{n-1}, \hat{x}_n, z_n)] \land
\cdots \land \\
&(\forall \hat{x}_2, \ldots, \hat{x}_n)(\exists z_2, \ldots, z_n)
[\phi(x_1, y_1, \hat{x}_2, z_2, \ldots, \hat{x}_n, z_n)]
],
\end{align*}
such that $\Phi \in \qbf \iff \Phi' \in \Pi_4\sat$.

The attempted proof is by induction on the number of variables in the formula.
We will not repeat the entire argument presented there, and we urge interested
readers to consult the original paper directly.
The gist of the purported inductive proof is as follows: We can swap the positions of quantifiers inside the original formula, and then we can detect in polynomial-time whether the formula's truth value has been affected by the changes.
We note, off the bat, two major issues with this approach: 
(1)~it does not leverage the implications of the statement of Theorem~1, and
(2)~the induction does not actually prove the logical equivalence.

Let us address (1)~first. The only mention made to Theorem~1 in the purported proof of
Theorem~2 is in the case where the number of variables, $m$, is greater or equal to 3.
The paper states that by ``taking off the first quantifier and checking 
both possible values for the first variable in [the] way we did in Theorem~1,
we come to the $m-1$ case''~\cite{sop:t:phpspace}.
It is worth reiterating that the attempted proof of 
Theorem~1 simply presents a version of the (well-known) recursive algorithm to decide $\qbf$ in polynomial space. Thus while this approach of ``eating'' variables one at a time
may help decide if the formula is true, it potentially uses an exponential amount of 
(nondeterministic) time,
and there is no clear passage in the paper that clarifies the use of that algorithm.

Let us now address~(2). In the purported inductive proof, the paper claims that for any
quantified boolean formula $\psi$ over two variables, $x$ and $y$, it holds that
$(\forall x)(\exists y)[\psi(x, y)] \equiv (\exists y)(\forall x)[\psi(x, y)]$ if and only if $\psi$
is neither the XOR function nor the negation of the XOR function, which is true. And thus, the argument in the paper
states that as long as there is no way for $\psi$ to be the XOR function (or its negation) when the values of all but two variables, with one being universally quantified over and the other being
existentially quantified over,
are fixed, then the induction holds. 
The 
additional clauses in
$\Phi'$ are meant to play a role in supporting this argument. However, the paper not only fails to explain how these additional clauses work and why they work, but it also seems to be missing cases.
Indeed, the paper states that the above check can be done in polynomial time since 
``there are [only] $n^2$ such formulas''~\cite{sop:t:phpspace}. 
We believe this was derived by selecting, from $\Phi$, one of the $n$ variables that are universally 
quantified over and one of the $n$ variables that are existential quantifier over. However,
missing from the argument is the fact that each of the remaining $2n-2$ variables can have one 
of two values (0 or 1), thus creating $n^22^{2n-2}$ possible formulas to check.
(It's worth noting that the above check can easily be done in
\emph{nondeterministic} polynomial time. However, the paper only states
``polynomial time,'' which, as is standard, implies ``deterministic polynomial time.'')
We thus conclude that the induction presented in that proof is incorrect.

We mention briefly in passing that there are other minor errors, but we bear them no attention
as they do not carry as much importance as the ones we pointed out. Additionally, in the final
parts of that proof, the paper mentions the formula's algebraic normal form as being crucial
to the argument and provides an example, but it is not clear why
the algebraic normal form is crucial. 
Thus we are not certain as to whether the author was hoping to achieve something different than what is being conveyed through the technical report.

\section{A Note on Skolemization}\label{s:skolem}

We add this section with the hopes that the use of Skolemization can be better explained. Our expectation is that the size of the Skolem functions must have a role to play in the proof.
Indeed, if the size of such functions is not polynomially-bounded, then it is unclear how any
skolemized formula can even be computed in polynomial space
(and be a useful approach in this context). On the other hand, if there is a polynomial
upper bound on the size of Skolem  functions, then one can show something shocking.

\begin{proposition}\label{prop:collapse}
If there is a polynomial $p : \naturalnumber \to \naturalnumberpositive$ such that, for each quantified boolean formula $\Phi$, the Skolemization of $\Phi$ produces no Skolem function of size greater
than $p(|\Phi|)$, then $\Sigma_2^p = \pspace$.
\end{proposition}
\begin{proof}[Proof]
The $\subseteq$ relationship is well-known, so it suffices to show that under the assumptions of the above statement, 
$\pspace \subseteq \Sigma_2^p$. We will do so by showing that $\qbf \in \Sigma_2^p$.
Let $p$ be a polynomial as defined by the proposition's statement.
Without loss of generality, let us assume that our input is a quantified boolean formula (as we can easily detect that
in polynomial time if it is not and immediately reject).
Fix an arbitrary boolean formula $\Phi$ with $n$ variables as our input. Let $E$ denote the set of variables in $\Phi$ that are
existentially quantified.
Since Skolemization preserves satisfiability, it
follows that $\Phi\in\qbf \iff $ there are $\|E\|$ Skolem functions such that for each assignment to the variables not in $E$, the boolean formula that results from replacing each variable in $E$ with the corresponding Skolem function is true under the current assignment. Because each Skolem function has size bounded by $p(|\Phi|)$, substituting the Skolem functions into $\Phi$ and checking the truth value of the resulting formula, given a specific assignment, can be computed in polynomial time. 
Per our definition of $\Sigma_2^p$
in Section~\ref{s:prelims}, this implies that $\qbf \in \Sigma_2^p$.
\end{proof}

This strengthens a result by Akshay et al.~\cite{aks-cha-goe-kul-sha:c:boolean-synthesis} who under similar assumptions conclude that $\Sigma_2^p = \Pi_2^p = \ph$ (by leveraging the Karp--Lipton Theorem).

\section{Conclusion}\label{s:conclusion}

In this critique, we pointed out the errors in 
``$\ph=\pspace$''~\cite{sop:t:phpspace} and concluded that Sopin's paper fails to show
that $\pspace$ and the polynomial hierarchy coincide.
The primary issue is that Sopin's paper does not account for a potentially exponential amount of work needed to perform one of its checks that it claims
can be done in polynomial time. We additionally find that the use of Skolemization in the attempted proof is not clear. Indeed, it would be shocking 
if a machine could compute the Skolemization using only a polynomial amount of space as that would yield
ground-breaking results (as proved in our Proposition~\ref{prop:collapse}).

\paragraph{Acknowledgements}
We would like to thank
Erin Gibson, 
David E. Narv\'{a}ez,
and 
Lane A. Hemaspaandra
for their helpful comments on prior drafts.
The authors are responsible for any remaining errors.

\bibliographystyle{alpha}
\bibliography{main}

\end{document}